\documentclass[letterpaper,11pt]{article}

\usepackage{verbatim}
\usepackage{latexsym}
\usepackage{amsmath,amssymb,amsthm}
\usepackage{fullpage}
\usepackage{epsfig,ifpdf,graphics}
\usepackage{times}
\usepackage{color}

\let\myPushQED=\pushQED
\let\myPopQED=\popQED
\newcommand{\myignore}[1]{}
\newenvironment{proof*}
  {\let\pushQED=\myignore\begin{proof}\let\pushQED=\myPushQED}
  {\def\popQED{}\end{proof}\let\popQED=\myPopQED}

\newenvironment{description*}%
  {\vspace{-1ex}\begin{description}%
    \setlength{\itemsep}{-0.5ex}%
    \setlength{\parsep}{0pt}}%
  {\end{description}}
\newenvironment{itemize*}%
  {\vspace{-1ex}\begin{itemize}%
    \setlength{\itemsep}{-0.5ex}%
    \setlength{\parsep}{0pt}}%
  {\end{itemize}}
\newenvironment{enumerate*}%
  {\vspace{-1ex}\begin{enumerate}%
    \setlength{\itemsep}{-0.5ex}%
    \setlength{\parsep}{0pt}}%
  {\end{enumerate}}

{\makeatletter
 \gdef\xxxmark{%
   \expandafter\ifx\csname @mpargs\endcsname\relax 
     \expandafter\ifx\csname @captype\endcsname\relax 
       \marginpar{xxx}
     \else
       xxx 
     \fi
   \else
     xxx 
   \fi}
 \gdef\xxx{\@ifnextchar[\xxx@lab\xxx@nolab}
 \long\gdef\xxx@lab[#1]#2{{\bf [\xxxmark #2 ---{\sc #1}]}}
 \long\gdef\xxx@nolab#1{{\bf [\xxxmark #1]}}
}

\newtheorem{theorem}{Theorem}
\newtheorem{lemma}[theorem]{Lemma}

\newtheorem{claim}[theorem]{Claim}

\newtheorem{observation}[theorem]{Observation}
\newtheorem{invariant}[theorem]{Invariant}

\newcommand{\eps}{\varepsilon}
\newcommand{\twodots}{\mathinner{\ldotp\ldotp}}
\newcommand{\proc}[1]{\textnormal{\scshape#1}}

\newcommand{\calD}{\mathcal{D}}
\newcommand{\Dyes}{\calD_{\mathrm{YES}}}
\newcommand{\Dno}{\calD_{\mathrm{NO}}}

\newcommand{\imax}{{i_{\max}}}
\newcommand{\imin}{{i_{\min}}}

\newcommand{\tmin}{t^{\min}}

\let\phi=\varphi
\newcommand{\E}{\mathbf{E}}

\newcommand{\foot}{\mathrm{Foot}}

\newcommand{\MSB}{\proc{high}}
\newcommand{\LSB}{\proc{low}}

\title{Using Hashing to Solve the Dictionary Problem \\(In External Memory)}

\author{
     John Iacono \\ NYU Poly
\and Mihai P\v{a}tra\c{s}cu \\ AT\&T Labs
}

\begin{document}

\maketitle

\begin{abstract}
We consider the dictionary problem in external memory and improve the
update time of the well-known \emph{buffer tree} by roughly a
logarithmic factor. For any $\lambda \ge \max \{ \lg\lg n, \log_{M/B}
(n/B) \}$, we can support updates in time $O(\frac{\lambda}{B})$ and
queries in sublogarithmic time, $O(\log_\lambda n)$. We also present a
lower bound in the cell-probe model showing that our data structure is
optimal.

In the RAM, hash tables have been used to solve the dictionary problem
faster than binary search for more than half a century. By contrast,
our data structure is the first to beat the comparison barrier in
external memory. Ours is also the first data structure to depart
convincingly from the \emph{indivisibility} paradigm.
\end{abstract}

\thispagestyle{empty}
\setcounter{page}{0}
\clearpage

\section{Introduction}

\paragraph{The case for buffer trees.}
The dictionary problem asks to maintain a set $S$ of up to $n$ keys
from the universe $U$, under insertions, deletions, and (exact) membership
queries. The keys may also have associated data (given at insert
time), which the queries must retrieve. Many types of hash tables can
solve the dictionary problem with constant time per operation, either
in expectation or with high probability. These solutions assume a
Random Access Machine (RAM) with words of $\Omega(\lg U)$ bits, which are
sufficient to store keys and pointers.

In today's computation environment, the \emph{external memory model} has
become an important alternative to the RAM. In this
model, it is assumed that there is an \emph{memory} which is partitioned into \emph{pages} of $B$ words. 
Accessing each page in memory takes unit time.  The
processor is also equipped with a \emph{cache} of $M$ words ($M/B$ pages),
which is free to access. The model can be applied at various levels,
depending on the size of the problem at hand. For instance, it can
model the interface between disk and main memory, or between main
memory and the CPU's cache.

Hash tables benefit only marginally from the external memory model: in
simple hash tables like chaining, the expected time per operation can
be decreased to $1 + 2^{-\Omega(B)}$ \cite{knuth-vol3}. Note that as
long as $M$ is smaller than $n$, the query time cannot go
significantly below $1$. However, the power of external memory lies in
the paradigm of \emph{buffering}, which permits significantly faster
updates.  In the most extreme case, if a data structure simply wants
to record a fast stream of updates without worrying about queries, it can do so with an amortized
complexity of $O(1/B) \ll 1$ per insertion: accumulate $B$ data items
in cache, and write them out at once into a page.

Buffer trees, introduced by Arge~\cite{arge03buffer}, are one of the
pillars of external memory data structures, along with
$B$-trees. Buffer trees allow insertions at a rate close to the ideal
$1/B$, while maintaining reasonably efficient (but superconstant)
queries. For instance, they allow update time $t_u = O(\frac{\lg
  n}{B})$ and query time $t_q = O(\lg n)$. More generally, they allow
the following range of trade-offs:

\begin{theorem}[Buffer trees \cite{arge03buffer}]\label{buffer}
Buffer trees support updates and queries with the following tradeoffs:
\begin{eqnarray}
t_u = O(\tfrac{\lambda}{B} \lg n) &\qquad&
  t_q = O(\log_\lambda n), \qquad \textrm{for } 2 \le \lambda \le B 
\label{eq:buff-low}  \\
t_u = O(\tfrac{1}{B} \log_\lambda n) & &
  t_q = O(\lambda \lg n), \qquad \textrm{for } 2 \le \lambda \le
  \tfrac{M}{B}.  \label{eq:buff-hi}
\end{eqnarray}
\end{theorem}

In these bounds and the rest of the paper, we make the following
reasonable and common assumptions about the parameters: $B \ge \lg n$; 
$M \ge B^{1+\eps}$ (tall cache assumption); $n \ge M^{1+\eps}$.

The motivation for fast (subconstant) updates in the external memory model
is quite strong. In applications where massive streams of data arrive
at a fast rate, the algorithm may need to operate close to the disk
transfer rate in order to keep up. On the other hand, we want reasonably
efficient data structuring to later find the proverbial needle in the
haystack.

A second motivation comes from the use of data structures in
algorithms. Sorting in external memory takes $O(\frac{n}{B} \log_M
n)$%
which is significantly sublinear for typical values of the parameters.
Achieving a bound close to this becomes the holy grail for many
algorithmic problems. Towards such an end, a data structure that
spends constant time per operation is of little relevance.

Finally, fast updates can translate into fast \emph{queries} in
realistic database scenarios. Database records typically contain many
fields, a subset of which are relevant in a typical queries. Thus, we
would like to maintain various indexes to help with different query
patterns. If updates are efficient, we can maintain more indexes, so
it is more likely that we can find a selective index for a future
query. (We note that this idea is precisely the premise of the
start-up company Tokutek, founded by Michael Bender, Mart\'in
Farach-Colton and Bradley Kuszmaul. By using buffer-tree technology to
support faster insertions in database indexes, the company is
reporting%
\footnote{{\tt http://tokutek.com/}.}
significant improvements in industrial database applications.)

\paragraph{The comparison/indivisibility barrier.}
In internal memory, hash tables can be traced back at least to 1953
\cite{knuth-vol3}. By contrast, in external memory the
state-of-the-art data structures (buffer trees, B-trees, many others
built upon them) are all comparison based!

The reason for using comparison-based data structures in external
memory seems more profound than in internal memory. Consider the
simple task of arranging $n$ keys in a desired order (permuting
data). The best known algorithm takes time $O(\min \{ n, \frac{n}{B}
\log_M n \})$: either implement the permutation ignoring the paging,
or use external memory sorting. Furthermore it has been known since
the seminal paper of Aggarwal and Vitter \cite{aggarwal88extmem} that
if the algorithm manipulates data items as indivisible atoms, this
bound is tight. It is often conjectured this this lower bound holds
for \emph{any} algorithm, not just those in the indivisible
model. This would imply that, whenever $B$ is large enough for the
internal-memory $O(n)$ solution to become irrelevant, a task as simple
as permuting becomes as hard as comparison-based sort.

While external-memory data structures do not need to be comparison
based, they naturally manipulate keys as indivisible objects. This
invariably leads to a comparison algorithm: the branching factors that
the data structure can achieve are related to the number of items in a
page, and such branching factors can be achieved even by simple
comparison-based algorithms. For problems such as dictionary,
predecessor search, or range reporting, the best known bounds are the
bounds of (comparison-based) B-trees or buffer trees, whenever $B$ is
large enough (and the external memory solution overtakes the RAM-based
solution). It is plausible to conjecture that this is an inherent
limitation of external memory data structures (in fact, such a
conjecture was put forth by \cite{yi10dagstuhl}).

Our work presents the first powerful use of hashing to solve the
external memory dictionary problem, and the first data structure to
depart significantly from the indivisibility paradigm. We obtain:

\begin{theorem}
For any $\max \{ \lg\lg n, \log_M n \} \le \lambda \le B$, we can
solve the dictionary problem by a Las Vegas data structure with update
time $t_u = O(\frac{\lambda}{B})$ and query time $t_q = O(\log_\lambda
n)$ with high probability.
\end{theorem}

\begin{figure}
\begin{tabular}{c|cc|cc|cc|}
&\multicolumn{2}{c|}{$\lambda=B^\eps$}
& \multicolumn{2}{c|}{$\lg \lg n \leq \log_M n$}
& \multicolumn{2}{c|}{$M=n^\eps$}\\
& $t_u$ & $t_q$
& $t_u$ & $t_q$
& $t_u$ & $t_q$
\\ \hline
Buffer Trees &  
$O(1/B^{1-\eps})$ & $O(\log_Bn)$ &
 $O(\frac{1}{B}\log_M n)$ & $\geq M^\eps \lg n$
&\parbox[c]{0.6in}{\begin{center}$O(\frac{\lg  n}{B})$\\$O(\frac{\lg \lg n}{B})$\end{center}}& 
\parbox[c]{0.8in}{\begin{center}$O(\lg n)$\\$2^{\Omega(\lg n/\lg \lg n)}$\end{center}}\\
Our Structure 
& $O(1/B^{1-\eps})$ & $O(\log_Bn)$ 
& $O(\frac{1}{B}\log_M n)$ & $O(\log M)$
& $O(\frac{\lg \lg n}{B})$ & $O(\frac{\lg n}{\lg \lg \lg n})$ \\
\end{tabular}

\caption{Selected update/query tradeoffs for buffer trees and our structure.}\label{tradeoffs}
\end{figure}

At the high end of the trade-off, for $\lambda = B^\eps$, we obtain
update time $O(1/B^{1-\eps})$ and query time $O(\log_B n)$ (See
Figure~\ref{tradeoffs}). This is the same as standard buffer trees.
Things are more interesting at the low end (fast updates), which is
the \emph{raison d'\^{e}tre} of buffer trees. Comparing to
Theorem~\ref{buffer}~\eqref{eq:buff-low}, our results are a
logarithmic improvement over buffer trees, which could achieve $t_u =
O(\tfrac{\lambda}{B} \lg n)$ and $t_q = O(\log_\lambda n)$.

Interestingly, the update time can be pushed very close to the ideal
disk transfer rate of $1/B$: we can obtain $t_u^{\min} = O(
\frac{1}{B} \cdot \max \{ \log_M n, \lg\lg n\})$.

Note that it is quite natural to restrict the update time to
$\Omega(\frac{1}{B} \log_M n)$. Unless one can break the permutation
bound, this is an inherent limitation of any data structure that has
some target order in which keys should settle after a long enough
presence in the data structure (be it the sorted order, or a
hash-based order). Since buffer trees work in the indivisible model,
they share this limitation. However, the buffer tree pays a
significant penalty in query time to achieve $t_u = O(\frac{1}{B}
\log_M n)$: from theorem~\ref{buffer}~\eqref{eq:buff-hi}, this requires $t_q \ge M^\eps \lg
n$, which is significantly more than polylogarithmic (for interesting
ranges of $M$). By contrast, our bound on the query time is still
(slightly) sublogarithmic.

If one assumes $M$ is fairly large (such as $n^\eps$), then $\lambda
\ge \lg\lg n$ becomes the bottleneck. This is an inherent limitation
of our new data structure. In this case, we can achieve $t_u =
O(\frac{\lg\lg n}{B})$ and $t_q = O(\lg n / \lg\lg\lg n)$. By
contrast, buffer trees naturally achieve $t_u = O(\frac{\lg n}{B})$
and $t_q = O(\lg n)$. With a comparable query time, our data structure
gets exponentially closer to the disk transfer rate for updates.
If we ask for $t_u = O(\frac{\lg\lg n}{B})$ in buffer trees, then, from
theorem~\ref{buffer}~\eqref{eq:buff-hi}, we must have a huge query time of $t_q =
2^{\Omega(\lg n / \lg\lg n)}$.

Our result suggests exciting possibilities in external memory data
structures. It is conceivable that, by abandoning the comparison and
indivisibility paradigms, long-standing running times of natural
problems such as predecessor search or range reporting can also be
improved.

\paragraph{Lower bounds.}
We complement our data structure with a lower bound that shows its
optimality:

\begin{theorem}
Let $\eps > 0$ be an arbitrary constant. Consider a data structure for
the membership problem with at most $n$ keys from the universe $[2n]$,
running in the cell-probe model with cells of $O(B\lg n)$ bits and a
state (cache) of $M$ bits. Assume $B\ge \lg n$ and $M \le
n^{1-\eps}$. The data structure may be randomized. Let $t_u$ be the
expected amortized update time, and $t_q$ be the query time. The query
need only be correct with probability $1- C_\eps$, where $C_\eps$ is
a constant depending on $\eps$.

If $t_u \le 1-\eps$, then $t_q = \Omega(\lg n / \lg(B\cdot t_u))$.
\end{theorem}

Remember that for any desired $t_u \ge t_u^{\min} = \frac{1}{B} \cdot
\max \{ \log_M n, \lg\lg n\}$, our new data structure obtained a query
time of $t_q = O(\lg n / \lg (B\cdot t_u))$. In other words, we have
shown an \emph{optimal} trade-off for any $t_u \ge t_u^{\min}$. We
conjecture that for $t_u = o(t_u^{\min})$, the query time cannot be
polylogarithmic in $n$.

Our lower bound holds for any reasonable cache size, $M \le
n^{1-\eps}$. One may wonder whether a better bound for smaller $M$ is
possible (e.g.~proving that for $t_u = o(\frac{1}{B} \log_B n)$, the
query time needs to be super-logarithmic). Unfortunately, proving this
may be very difficult. If sorting $n$ keys in external memory were to
take time $O(n/B)$, then our data structure will work for any $t_u \ge
\Omega(\lg\lg n / B)$, regardless of $M$. Thus, a better lower bound
for small cache size would imply that sorting requires superlinear
time (and, in particular, a superlinear circuit lower bound, which
would be a very significant progress in complexity theory).

Remember that update time below $1/B$ is unattainable, regardless of
the query time. Thus, the remaining gap in understanding membership is
in the range $t_u \in \big[\frac{1}{B}, \frac{\lg\lg n}{B} \big]$.

At the high end of our trade-off, we see a sharp discontinuity between
internal memory solutions (hash tables with $t_u = t_q \approx 1$) and
buffer trees. For any $t_u < 1-\eps$, the query time blows up to
$\Omega(\log_B n)$.

The lower bound works in the strongest possible conditions: it holds
even for membership and allows Monte Carlo randomization. Note that
the error probability can be made an arbitrarily small constant by
$O(1)$ parallel constructions of the data structure. However, since we
want a clean phase transition between $t_u = 1-\eps$ and $t_u = 1$, we
mandate a fixed constant bound on the error.

The first cell-probe lower bounds for external-memory membership was
by Yi and Zhang~\cite{yi10buffer} in SODA'10. Essentially, they show
that if $t_u \le 0.9$, then $t_q \ge 1.01$. This bound was
significantly strengthened by Verbin and Zhang~\cite{verbin10buffer}
in STOC'10. They showed that for any $t_u \le 1-\eps$, then $t_q =
\Omega(\log_B n)$.

This bound is recovered as the most extreme point on our trade-off
curve. However, our proof is significantly simpler than that of Verbin
and Zhang. We also note that the technique of \cite{verbin10buffer}
does not yield better query bounds for fast updates. A lower bound for
small $t_u$ is particularly interesting given our new data structure
and the regime in which buffer trees are most appealing.

For update time $t_u \ll \frac{\lg n}{B}$, our lower bound even beats
the best known comparison lower bound. This was shown by Brodal and
Fagerberg \cite{brodal03buffer} in SODA'03, and states that $t_q =
\Omega(\lg n / \log(t_u B \lg n))$. On the other hand, in the
comparison model, it was possible to show \cite{brodal03buffer} that
one cannot take $t_u \ll \frac{1}{B} \log_M n$ (below the permutation
barrier) without a significant penalty in query time: $t_q \ge
n^{\Omega(1)}$.


\section{Upper Bound}

Our data structure is presented in a number of levels. First, in
section~\ref{sec:shrinkage}, we describe how we can map word-sized
keys into keys with $O(\log n)$ bits; a brief discussion of how
deletions can be handled using insertions also appears in this
preliminary high-level section.

In section~\ref{sec:gadget}, we then proceed to describe the core
component of our structure, called a ``\emph{gadget}.'' The gadget is
defined recursively, and a description of the recursive implementation
of the operations is presented in Section~\ref{sec:recimp}. The small
non-recursive gadget at the base of the recursion is nontrivial enough
to merit a separate description appearing in
Section~\ref{sec:base}. We present a high-level analysis of the gadget
in Section~\ref{s:analysis}; some probability arguments about the size
of gadgets and key-value-collisions are needed, which are isolated in
Section~\ref{s:prob}.

Our gadget has certain size limitations, in that it can only be used
as presented if it fits into cache. So, for larger data we use as
global recursive structure a variant of buffer trees, where we switch
to our gadgets when at levels of the recursion where the size
requirements of out gadget are met. We elaborate on this idea in
Section~\ref{s:global}, which completes the description of our
structure.

To summarize: after preprocessing (\S \ref{sec:shrinkage}), a
buffer-tree based structure is used, where the leaves are gadgets (\S
\ref{s:global}); gadgets are defined recursively (\S \ref{sec:gadget})
with a non-trivial base structure (\S \ref{sec:base}).

\subsection{Preliminary key shrinkage and deletions}\label{sec:shrinkage}

As a warm-up, we show how we can assume that keys and associated data
have $O(\lg n)$ bits. The data structure can simply log keys in an
array, ordered by insertion time. In addition, it hashes keys to the
universe $[n^2]$, and inserts each key into a buffer tree with an
index into the time-ordered array (of $O(\lg n)$ bits) as associated
data. A buffer-tree query may return several items with the same hash
value. Using the associated pointers, we can inspect the true value of
the key in the logging array. Each takes one memory access, but we
know that there are only $O(1)$ false positives with high probability
(w.h.p.), with a good enough hash function.

Deletions can be handled in a black-box fashion: add the key to the
logging array with a special associated value that indicates a delete,
and insert it into the buffer tree normally. Throughout the paper, we
will consider buffer trees with an ``overwrite'' semantics: if the
same key is inserted multiple times, with different associated values,
only the last one is relevant to the query. Thus, the query will
return a pointer to the deletion entry in the log. After $O(n)$
deletions we perform global rebuilding of the data structure to keep
the array bounded.

From now on, we assume the keys have $O(\lg n)$ bits. Let $b =
\Omega(B\lg n)$ be the number of bits in a page.

\subsection{Gadgets} \label{sec:gadget}

The fundamental building block of our structure is called a
\emph{gadget}.  A $t$-gadget stores a multiset $S$ from $([b]\times [t]
\times [t])\times [t^3]$. We refer to the components of a tuple
$x=((p,d,s),b) \in S$ as:
\begin{itemize*}
\item the \emph{page hash}, $p\in[b]$;
\item the \emph{distribution hash}, $d \in [t]$;
\item the \emph{shadow hash}, $s \in [t]$;
\item the \emph{backpointer}, $b \in [t^3]$. Taken together, the page,
  distribution, and shadow hash codes are treated as a key, with the
  backpointer being associated data.
\end{itemize*}

\paragraph{Operations.}

A $t$-gadget stores a multiset supporting two operations:
\begin{description*}
\item[$\proc{Bulk-Insert}(T)$:] Insert a multiset $T \subseteq ([b]\times
  [t]^2)\times [t^3]$ into the data structure ($S \leftarrow S \cup T$). The
  multiset is presented packed into  $O(\lceil |T| / \frac{b}{\lg t}
  \rceil)$ pages.

\item[$\proc{Query}(x)$:] Given a key $x \in [b]\times [t]^2$, return
  the (possibly empty) list of the backpointer values of all elements
  of $S$ with key value $x$. We aim for time bounds proportional to
  the number of occurrences of elements with key value $x$ in the
  multiset.

\end{description*}

\paragraph{Capacity invariants.}
Our construction will guarantee that the the number of keys stored in
a $t$-gadget is $|S| = O(bt)$ with high probability.  A $t$-gadget
will occupy $O(|S| / \frac{b}{\lg t}) = O(t\lg t)$ pages thanks to the
use of a succinct encoding. At the beginning of any operation, there
is no guarantee that any part of the $t$-gadget is in cache.
However, we will only use gadgets that can completely fit in cache,
i.e.~the entire gadget can be loaded in cache by the update algorithm
if desired.

\paragraph{A recursive construction.}
At a high level, our data structure follows the same recursive
construction used in the van Emde Boas layout
two types of $t$-gadgets: the \emph{recursive $t$-gadget}, and the
\emph{base $t$-gadget} which is used as a base case for small $t$. The
description of the base $t$-gadget is deferred to \ref{sec:base}. Here
we define the recursive $t$-gadget.

Recursive $t$-gadgets contain the following components:

\begin{itemize}
\item The \emph{log}: An array containing all the elements of $S$ in
  the order of insertion. The last block of the array is the only one
  which may be partially full and is referred to as the \emph{tail
    block} of the log.

\item The \emph{top gadget $g^T$}: A a recursive $\sqrt{t}$-gadget.

\item The \emph{bottom gadgets $g^B_i$}: An array of $\sqrt{t}$~
  $\sqrt{t}$-gadgets.

\end{itemize}

All elements of $S$ are stored in the log; furthermore, all elements
of $S$ except for those in the tail block will have a truncated
representation of their keys recursively stored in either the top
gadget or one of the bottom gadgets. Formally:

\begin{invariant} \label{inv}
Let $\MSB(\cdot)$ and $\LSB(\cdot)$ refer to the most and least
significant half of the bits of their parameter. Given an element
$x=((p,d,s),b)$ exactly one of the following holds:

\begin{enumerate}

\item \label{c1} Element $x$ is stored in the tail block of the log.

\item \label{c2} Element $x$ is stored in block $i$ of the log and
  $((p,\MSB(d),\MSB(s)),i)$ is stored in the top gadget, $g^T$. The
  element $((p,\MSB(d),\MSB(s)),i)$ is called the
  ``\emph{top-compressed}'' key.

\item \label{c3} Element $x$ is stored in block $i$ of the log and
  $((p,\LSB(d),\LSB(s)),i)$ is stored in $g^B_{\MSB(d)}$. The element
  $((p,\LSB(d),\LSB(s)),i)$ is called the ``\emph{bottom-compressed}''
  key.
\end{enumerate}

\end{invariant}

Given an element $x$ stored in a $t$-gadget, the top-compressed
version of $x$ consists of the entire page hash of $x$, and the
$\frac{1}{2}\lg t$ higher-order bits of the distribution hash and
shadow hash.  The \emph{bottom-compressed} version of $x$ is
constructed analogously using the page hash and and lower-order bits
of the distribution hash and shadow hash. In both cases the
backpointer indicates the page containing $x$ in the log of the
$t$-gadget. These compressed elements meet the size requirements for
elements that can be bulk-inserted into $\sqrt{t}$-gadget; i.e.~the
compressed key by definition is an element of $[b]\times[\sqrt{t}]^2$.
Since $|S| = \Theta(bt)$, there are at most $\Theta(t\lg t)$ pages in
the log, so the backpointer is easily within the required $[t^3]$.

The backpointer serves the following purpose: given a top or bottom
compression of some element $x$ relative to a specific and known
$t$-gadget, $x$ can be determined in $O(1)$ time by simply using the
backpointer of the compression as an index into the log of the
$t$-gadget. In this way the bits removed from the distribution and
shadow hashes of of $x$ to form its compression can be restored; as we
will see, this allows the support of ``uncompression'' as recursive
queries return.

\subsubsection{Implementation of gadget operations}\label{sec:recimp}


\paragraph{Query.}
A $\proc{Query}(x)$ operation proceeds as follows, where $x=(p,d,s)$
is the current recursive compression of the original hashed key:
\begin{itemize*}
\item Inspect the tail block of the log, and retrieve the associated
  backpointers of all occurrences of $x$ from there. This takes one
  block read, and returns all the backpointers of the data satisfying
  case~\ref{c1} of invariant~\ref{inv}.

\item Recursively call \proc{Query} in the top gadget $g^T$ with the
  top-compressed key $(p,\MSB(d),\MSB(s))$. The top gadget will return a
  set of backpointers $\{p_1, p_2, \ldots\}$, which are indexes into
  the log of the current gadget. The set of items from the log
  includes all data satisfying case~\ref{c2} of invariant~\ref{inv},
  plus possibly some {\bf false positives}, where the compressed key
  matches the query key from the perspective of the top gadget, but
  not the full key is different. For any result returned by the top
  gadget, the query inspects the key in the log (taking constant time
  for the pointer access) and verifies the lower halves of the
  distribution and shadow hash codes, $\LSB(d)$ and $\LSB(s)$. If any of
  these differ, the result is a false positive and is
  discarded. Otherwise, the result is returned along with the original
  backpointer, retrieved from the log. We will later need to bound the
  number of false positives induced by hashing and compression.

\item Recursively call \proc{Query} in the bottom gadget
  $g^B_{\MSB(d)}$ with the compressed key $(p,\LSB(d),\LSB(s))$. This
  returns all data satisfying case~\ref{c3} of Invariant~\ref{inv},
  together with some false positives that may be introduced by
  trimming the shadow hash code. For every recursive result, the query
  accesses the appropriate page in the log through the returned
  backpointer, and verifies that the lower half of the shadow hash
  code $\LSB(s)$ matches the key. If not, the result is discarded as a
  false positive. In case of a match, the original backpointer from
  the log is returned to the parent.
\end{itemize*}

Observe that recursing in the bottom gadget cannot introduce a false
positive due to the distribution hash, since it is only keys with the
same top bits as $x$ that appear in the gadget.

\paragraph{Bulk-Insert.}
An insertion proceeds as follows:

\begin{enumerate}
\item Add the inserted items at the end of the log. \label{s1}

\item If one or more than one blocks are filled as a result of step
  \ref{s1}., the recursive top compressed representation of the data
  in the newly filled blocks is computed and the resultant
  top-compressed data is bulk-inserted into the top gadget. We call
  this a \emph{little flush}. \label{s2}

\item If the top gadget contains $b \sqrt{t}$ keys as a result of
  (\ref{s2}), it is declared to be \emph{full}. It is then
  ``destroyed'' (initialized to an empty state) and all keys
  previously stored in the top gadget are bottom-compressed and
  inserted into the appropriate bottom gadget. We call this a
  \emph{big flush}. To efficiently implement this operation, the data
  to be flushed is copied from the log, where it appears contiguously
  in uncompressed form, into cache, where it all fits according to the
  capacity invariant. There the data is bucketed (for free using any
  sorting algorithm, since we are in internal memory) into $\sqrt{t}$
  groups depending on the $\MSB(d)$ field which indicates which
  recursive gadget it should be inserted into. Once the bucketing is
  complete, the data is converted into the appropriate compressed
  form, $\sqrt{t}$ recursive \proc{Bulk-Insert} operations are
  executed in the bottom gadgets.
\label{s3}
\end{enumerate}

Note that this procedure enforces Invariant~\ref{inv}, by ensuring
each item is either in the tail block of the log, in the top gadget,
or in a bottom gadget. The capacity requirement of the top gadget is
explicitly enforced. The capacity requirement of any bottom gadget
holds w.h.p; this is argued in Section \ref{s:prob}.

\subsubsection{Base case} \label{sec:base}
We switch to the base case of the recursion when $t\le \tmin$, for a
parameter $\tmin$ to be determined (see \S\ref{s:global}). To achieve
the full range of our trade-offs, we need to use a different
non-recursive construction for these small gadgets. Such a gadget
maintains a single \emph{buffer page} with the last $\le \frac{b}{\lg
  t}$ inserted keys.  The rest of the keys are simply stored in a hash
table (e.g.~collision chaining) addressed by the \emph{page
  hash}. This is the one place where the page hash is used. With a
succinct representation, the table occupies $O(|S| / \frac{b}{\lg t})
= O(t \lg t)$ pages.

$\proc{Query}(x)$ inspects the buffer page and only one page of the
hash table w.h.p. (since we have assumed $B = \Omega(\lg n)$, and the
maximal chain is $O(\lg n / \lg\lg n)$ w.h.p.). Thus, a query takes
time $O(1)$. $\proc{Update}$ simply appends keys to the buffer
page. When the buffer page fills, all keys are inserted into the hash
table. This operation may need to touch all $O(t\lg t)$ pages, since
the new keys are likely to have hash codes that are spread out;
however the cost can not exceed $O(t \lg t)$ since the capacity
invariants ensure the whole gadget fits in memory.

\subsubsection{Analysis of the $t$-gagdet} \label{s:analysis}

In this section, we analyze the performance of gadgets, delaying the
probabilistic analysis to the next section.

\paragraph{Space usage.} 
Though a key stored in a gadget may appear in $O(\lg\lg n)$ recursive
gadgets, the repeated compression results in the space used by all
compressed occurrences of the key is dominated by the top-level
representation of the key.  Formally, a single key in a $t$ gadget
occupies $O(\lg t)$ bits of space and may appear recursively in at
most one $\sqrt{t}$-gadget. Thus, the space per key is given by the
recurrence $S(t)\leq S(\sqrt{t})+O(\lg t)$, which solves to $O(\lg t)$
bits. Overall, a $t$-gadget storing $n$ keys uses $O(n)$ words of
space.

\paragraph{Update cost.}
Over its lifetime in a $t$-gadget, an element will be appended to the
log once, participate in a little flush (being inserted into the top
gadget) at most once, and participate in a a big flush (being moved
into a bottom gadget) at most once.

We begin our analysis with the cost of \proc{Bulk-Insert}, excluding
recursive calls. A key participates in a \proc{Bulk-Insert} operation
if it is one of the inserted items or participates in a big or little
flush. If $k$ is the number of participating items, the actual running
time excluding recursive calls, is $O(1+\frac{k\lg t}{b})$. Since
$\frac{k\lg t}{b}$ is the (fractional) number of blocks occupied by a
single element in a $t$-gadget, this is the fastest possible (linear
time). Let us briefly explain how this is achieved for each of the
three steps presented in the description of the \proc{bulk-insert}
operation. The first step, inserting at the end of the buffer, can be
done efficiently since we required that the inserted data is already
presented packed into pages. In the little flush, one or more blocks
of data from the log need to be copied and converted into
top-compressed form and delivered to single recursive
\proc{bulk-insert}; this can be done with a simple scan. The big
flush, by definition is performed when the top $\sqrt{t}$-buffer is
full, containing $b\sqrt{t}$ keys; the actual cost (including the
calling of but excluding the execution of the recursive bulk
insertions) is $\Theta(\sqrt{t} \log \sqrt{t}+\sqrt{t})$, with the
$\sqrt{t}$ being a lower-order term due to the $\sqrt{t}$ recursive
calls to \proc{Bulk-Insert}.

We now bound the recursive cost by amortizing. To cover the constant
additive term, we assign an amortized $O(1)$ credit to every
\proc{Bulk-Insert} operation. The credit for the recursive call in the
top gadget can be paid because a recursive call is only made when we
fill a page of the log. The credit for the recursive calls in the
bottom gadgets is a lower order term compared to sorting the entire
top gadget.

Now we are left with a cost of $O(\frac{k\lg t}{b})$ for a
\proc{Bulk-Insert} operation in which $k$ keys participate. This
translates into an amortized cost of $O(\frac{\lg t}{b})$ per key. Let
$U(t)$ be the total cost charged to a key by a $t$-gadget, including
recursive calls. This is described by the recurrence:

\[ U(t) = O(\tfrac{\lg t}{b}) + 2 \cdot U(\sqrt{t});
\qquad\qquad 
U(t \leq \tmin) = O(\tfrac{t\lg t}{b})
\]

Observe that on each level of the recurrence, we have $2^i$ terms of
$O(\frac{1}{b} \lg(t^{-2^i}))$, i.e.~a constant total cost of
$O(\frac{\lg t}{b})$ per level. This is very intuitive: at each level,
our key is broken into many components, but their total size stays
$\lg t$ bits. Since the cost is proportional to the bit complexity of
the key, the cost of each level is constant; e.g.~at the top level you
recurse twice on keys of half the size. This property of the data
structure is the most crucial element in obtaining our upper bound. It
only possible due to our compression of the keys; without compression,
i.e.~without violating indivisibly, the cost would increase
geometrically at each level instead of remaining unchanged.

The recursion for $U(t)$ solves as follows: the recursion has
$O(\lg\lg t)$ levels at a cost of $O(\frac{\lg t}{b})$ per level. In
the base case, the recursion has $\lg t / \lg \tmin$ leaves, each of
cost $O(\frac{1}{b} \tmin \lg \tmin)$. Thus the total cost is $U(t)
= \frac{\lg t}{b} \cdot O(\lg\lg t + \tmin)$.

\paragraph{Query cost.}
The query time is proportional to the number of (true positive)
results returned. Let $Q(t)$ be the cost per result of a query in a
$t$-gadget. We first note that $Q(t)$ is at least $1$, because any
result has to be looked up in the log, in order to check whether it is a
true positive. The query cost is proportional to the
number of gadgets traversed, and is described by the easy recursion:
\[ Q(t) = 1 + 2\cdot Q(\sqrt{t});
\qquad \qquad Q(\tmin) = 1 \] 
The number of gadgets grows exponentially with the level, and is
dominated by the base case. The total cost is therefore $Q(t) = O(\lg
t / \lg \tmin)$.

For each false positive encountered throughout the query, there is an
additional cost of at most $O(\lg\frac{\lg t}{\lg \tmin})$. Indeed,
each key is in at most $O(\lg\frac{\lg t}{\lg\tmin})$ levels at a
time, and we need to do $O(1)$ work per level: we go to the
appropriate page in the log to verify the identity of the key and
retrieve its data. We will show later that the total overhead due to
false positives is $O(Q(t))$ w.h.p.

\subsubsection{Probabilistic Analysis} \label{s:prob}

\paragraph{Capacity bounds.}
We first prove that no $t$-gadget $g$ receives more than $O(bt)$ keys
w.h.p. Note that shadow and page hashes are irrelevant to this
question, and we only need to analyze distribution hashes. For now,
assume that our hashing is truly random.

Let $g'$ be the lowest ancestor of $g$ in the recursion tree which is
a top-recursive gadget of its parent; say $g'$ is a $t'$-gadget. We
conventionally interpret the root gadget to be a top-recursive gadget,
so $g'$ is always defined. Note that $g'$ could be $g$. Remember that
\proc{Bulk-Insert} enforces a worst-case capacity bound on any top
gadget by the big flush operation, so the number of keys of $t'$ is at
most $b\cdot t'$ in the worst case. A key from a $t'$-gadget ends up
in a specific bottom gadget only if the first half of its distribution
hash matches the identity of the bottom gadget. Recursively, keys that
end up in a specific grandchild $\sqrt[4]{t}$-gadget have the same
prefix of $\frac{3}{4} \lg (t')$ bits of the distribution hash. Since
$g'$ is the lowest ancestor of $g$ that is a top-recursive gadget, all
keys of $g'$ that end up in $g$ do so through bottom recursion,
i.e.~they will all have a common prefix of $\lg(t') - \lg t$ bits.
Therefore, analyzing the number of keys of $g'$ that end up in $g$ is
a standard balls-in-bins problem with the $bt'$ balls of $g'$ being
distributed uniformly into $\frac{t'}{t}$ bins. The expected number of
balls landing in gadget $g$ is $tb$. Since $b = \omega(\lg n)$, the
Chernoff bound says that we have at most $O(tb)$ keys in the bin with
high probability in $n$.

\paragraph{False positives during \proc{Query}.}
We now switch to analyzing the number of false positives encountered
by \proc{Query}(x).  We count a false positive only once, in the first
level of recursion where it is introduced. As noted above, a false
positive introduced in a $t$-gadget induces an additive cost of
$O(\lg\frac{\lg t}{\lg \tmin})$ on the query time.

We claim that for any $t$, the number of the false positives
introduced in all $t$-gadgets that the query traverses is $O(\log_t
n)$ w.h.p. Thus, the total cost on one level of the recursion is
$O(\frac{\lg n}{\lg t} \cdot \lg\frac{\lg t}{\lg \tmin})$. At the
$i$-th level of the recursion, we have $\lg t = 2^i\lg \tmin$, so the
total cost is:
\[ \sum_i O\Big( \frac{\lg n}{\lg(2^i \lg tmin)} \cdot
                  \lg\frac{\lg(2^i \lg\tmin)}{\lg \tmin} \Big) 
= O\Big(\frac{\lg n}{\lg \tmin} \Big) \sum i 2^{-i} 
= O\Big( \frac{\lg n}{\lg \tmin} \Big)
= O(Q(t)).\]

Thus, the total cost due to false positives is $O(Q(t))$ w.h.p.

We must now prove our claim that the false positives introduced in all
$t$-gadgets on the query path is $O(\log_t n)$ w.h.p. There are
$\log_t n$ $t$-gadgets on the path of the query, and each cares about
a disjoint interval of $\lg t$ bits of the distribution and shadow
hash codes. For the analysis, we imagine fixing a growing prefix of
the distribution and shadow hashes, in increments of $\lg t$ bits. At
every step, the fixing so far decides which keys land in the next
$t$-gadget. Among these, only the most recent $\sqrt{t} \cdot b$ can
be in the top gadget at query time. One of these keys is a false
positive iff it is different from the query, yet its page hash and the
top half of its distribution and shadow hash codes match the query
key. These hash codes consist of $\lg b + 2\lg \sqrt{t}$ random bits,
so this event happens with probability $\frac{1}{bt}$ independently
for each of the $\sqrt{t} \cdot b$ keys. For bottom compression, a key
can only be a false positive iff it is different from the query, yet
its page hash, full distribution hash, and the bottom half of the
shadow hash match the query. These hash codes have $\lg b +
\frac{3}{2} \lg \sqrt{t}$ random bits together, so a key is a false
positive in a bottom recursion with probability $\frac{1}{b
  t^{3/2}}$. Over the path of the query through $\log_t n$ different
$t$-gadgets, there are $b\sqrt{t} \cdot \log_t n$ keys that could
become false positives through top recursion (each independently with
probability $\frac{1}{bt}$), and w.h.p.~$bt \cdot \log_t n$ keys that
could become false positive through bottom recursion (each
independently with probability $\frac{1}{b t^{3/2}}$). The expected
number of false positives among all $t$-gadgets is therefore
$O(\frac{\log_t n}{\sqrt{t}})$. The Chernoff bound (in the
Poisson-type regime) shows that the number of false positives does not
exceed $O(\log_t n)$ with probability $(1/\sqrt{t})^{\Omega(\log_t n)}
= 2^{-\Omega(\lg n)}$, i.e.~with high probability in $n$.

Since our analysis only relies on Chernoff bounds, it holds even with
weaker hash functions that satisfy such bounds. The main requirement
for the hash function is that it satisfy Chernoff-type concentration
even among the set of keys that share a certain prefix of the hash
code. This is true of $k$-independent hash functions, since a
$k$-independent distribution remains $k$-independent if we fix a
prefix of the bits of the outcome. Thus, a $\Theta(\lg n)$-independent
hash function suffices for our data structure, considering Chernoff
bounds with limited independence given by \cite{schmidt95chernoff}.
These hash functions can be represented in $O(\lg n)$ words, which
fits in cache, and therefore can be evaluated without I/O cost.

\subsection{Putting Everything Together} \label{s:global}

We now describe how gadgets are used to make a dictionary structure
that supports \proc{Insert} and \proc{Search} operations. Recall that
the capacity invariant of our gadget requires that that the total size
of a gadget is smaller than the cache size, so a single gadget in and
of itself cannot be the target data structure.

Our data structure is globally organized like the original buffer tree,
except that each node is implemented by a gadget to support fast
queries. At the global level, only comparisons are needed.

Formally, we have a tree with branching
factor $M^\eps$. Each node can store up to $M$ keys. These keys are stored in a plain array, but also in an $\frac{M}{b}$-gadget. The
gadgets are independent random constructions (with different
i.i.d.~hash functions). The $O(\log n)$-bit key values are partitioned into distribution, shadow, and block hash fields when inserting and searching the gadgets in the nodes.  When the capacity of the node fills, its keys
are distributed to its children. Since $M\ge B^{1+\eps}$ (the tall
cache assumption), this distribution is I/O efficient, costing
$O(\frac{1}{B})$ per key.

The total cost for inserting a key is:
\[ O(\tfrac{1}{B} \log_M n) + U(\tfrac{M}{b}) \cdot \log_M n
=  O(\tfrac{\log_M n}{B}) + O(\log_M n) \cdot
  \tfrac{\lg M}{B\lg n} \cdot (\lg\lg M + \tmin \lg \tmin)
\]
Thus $t_u = O(\tfrac{1}{B}) \cdot \big( \log_M n + \lg\lg M + \tmin
\lg \tmin \big)$.  Note that $\log_M n + \lg\lg M = \Theta(\log_M n +
\lg\lg n)$.  For any $\lambda\ge \log_M n + \lg\lg n$, we can achieve
$t_u = O(\lambda / B)$ by setting $\tmin \lg \tmin = \lambda$.

The cost of querying a key is $Q(\frac{M}{b}) \cdot O(\log_M n)
= O( \frac{\lg M}{\lg \tmin} \cdot \log_M n) = O(\frac{\lg n}{\lg
  \tmin}) = O(\log_\lambda n)$.

The total space of our data structure is linear. Each key is stored in
only one top-level gadget at a time, and, as we argued
in~\ref{s:analysis}, these are have linear size (in words).

\bibliographystyle{alpha} 
\bibliography{general}

\appendix

\section{Lower Bounds}


Formally, our model of computation is defined as follows. The update
and query algorithms execute on a processor with $M$ bits of state,
which is preserved from one operation to the next. The memory is an
array of cells of $O(B\lg n)$ bits each, and allows random access in
constant time. The processor is nonuniform: at each step, it decides
on a memory cell as an arbitrary function of its state. It can either
write the cell (with a value chosen as an arbitrary function of the
state), or read the cell and update the state as an arbitrary function
of the cell contents.

We first define our hard distribution. Choose $S$ to be a random set
of $n$ keys from the universe $[2n]$, and insert the keys of $S$ in
random order. The sequence contains a single query, which appears at a
random position in the second half of the stream of operations. In
other words, we pick $t \in [\frac{n}{2}, n]$ and place the query
after the $t$-th update (say, at time $t + \frac{1}{2}$). 

Looking back in time from the query, we group the updates into epochs
of $\lambda^i$ updates, where $\lambda$ is a parameter to be
determined.  Let $S_i \subset S$ be the keys inserted in epoch $i$;
$|S_i| = \lambda^i$. Let $\imax$ be the largest value such that
$\sum_{i=1}^\imax \lambda^i \le \frac{n}{2}$ (remember that there are
at least $n/2$ updates before the query). We always construct $\imax$
epochs. Let $\imin$ be the smallest value such that $\lambda^{\imin-1}
\ge M$; remember that $M < n^{1-\eps}$ was the cache size in bits. Our
lower bound will generally ignore epochs below $\imin$, since most of
those elements may be present in the cache. Note $\imax - \imin \ge
\eps \log_\lambda n - O(1)$.

We have not yet specified the distribution of the query. Let $\Dno$ be
the distribution in which the query is chosen uniformly at random from
$[2n] \setminus S$.  Let $\calD_i$ be the distribution in which the
query is uniformly chosen among $S_i$. Then, $\Dyes = \frac{1}{\imax -
  \imin + 1} (\calD_\imin + \dots + \calD_\imax)$. Observe that
elements in smaller epochs have a much higher probability of being the
query. We will prove a lower bound on the mixture $\frac{1}{2} (\Dyes
+ \Dno)$. Since we are working under a distribution, we may assume the
data structure is deterministic by fixing its random coins
(nonuniformly).

\begin{observation}
Assume the data structure has error $C_\eps$ on $\frac{1}{2} (\Dyes +
\Dno)$. Then the error on $\Dno$ is at most $2C_\eps$. At least half
the choices of $i \in \{ \imin, \dots, \imax \}$ are \emph{good} in the
sense that the error on $\calD_i$ is at most $4 C_\eps$.
\end{observation}

We will now present the high-level structure of our proof. We focus on
some epoch $i$ and try to prove that a random query in $\Dno$ reads
$\Omega(1)$ cells that are somehow ``related'' to epoch $i$.  To
achieve this, we consider the following encoding problem: for $k =
\lambda^i$, the problem is to send a random bit vector $A[1\twodots
  2k]$ with exactly $k$ ones. The entropy of $A$ is $\log_2
\binom{2k}{k} = 2k - O(\lg k)$ bits.

The encoder constructs a membership instance based on its array $A$
and public coins (which the decoder can also see). It then runs the
data structure on this instance, and sends a message whose size
depends on the efficiency of the data structure. Comparing the message
size to the entropy lower bound, we get an average case cell-probe
lower bound.

The encoder constructs a data structure instance as follows:
\begin{itemize*}
\item Pick the position of the query by public coins. Also pick
  $S^\star = S \setminus S_i$ by public coins.

\item Pick $Q = \{ q_1, \dots, q_{2k} \}$ a random set of $k$ keys
  from $[2n] \setminus S^\star$, also by public coins.

\item Let $S_i = \{ q_j \mid A[j] = 1\}$, i.e.~the message is encoded
  in the choice of $S_i$ out of $Q$.

\item Run every query in $Q$ on the memory snapshot just before the
  position chosen in step 1.
\end{itemize*}

The idea is to send a message, depending on the actions of the data
structure, that will also allow the decoder to simulate the
queries. Note that, up to the small probability of a query error, this
allows the decoder to recover $A$: $A[j]=1$ iff $q_j \in S$.

It is crucial to note that the above process generates an $S$ that is
uniform inside $[2n]$. Furthermore, each $q_j$ with $A[j]=0$ is
uniformly distributed over $[2n] \setminus S$. That means that for
each negative query, the data structure is being simulated on
$\Dno$. For each positive query, the data structure is being simulated
on $\calD_i$.

Let $R_i$ be the cells read during epoch $i$, and $W_i$ be the cells
written during epoch $i$ (these are random variables). We will use the
convenient notation $W_{<i} = \bigcup_{j=0}^{i-1} W_j$. 

The encoding algorithm will require different ideas for $t_u = 1-\eps$
and $t_u=o(1)$.

\subsection{Constant Update Time}

We now give a much simpler proof for the lower bound of
\cite{verbin10buffer}: for any $t_u \le 1-\eps$, the query time is
$t_q = \Omega(\log_B n)$. 

We first note that the decoder can compute the snapshot of the memory
before the beginning of epoch $i$, by simply simulating the data
structure ($S_{>i}$ was chosen by public coins). The message of the
encoder will contain the cache right before the query, and the address
and contents of $W_{<i}$. To bound $|W_{<i}|$ we use:

\begin{observation}  \label{obs:wi}
For any epoch $i$, $\E[|W_i|] \le \lambda^i t_u$.
\end{observation}

\begin{proof}
We have a sequence of $n$ updates, whose average expected cost (by
amortization) is $t_u$ cell probes. Epochs $i$ is an interval of
$\lambda^i$ updates. Since the query is placed randomly, this interval
is shifted randomly, so its expected cost is $\lambda^i t_u$.
\end{proof}

We start by defining a crucial notion: the footprint $\foot_i(q)$ of
query $q$, relative to epoch $i$. Informally, these are the cells that
query $q$ would read if the decoder simulated it ``to the best of its
knowledge.'' Formally, we simulate $q$, but whenever it reads a cell
from $W_i \setminus W_{<i}$, we feed into the simulation the value of
the cell from before epoch $i$. Of course, the simulation may be
incorrect if at least one cell from $W_i \setminus W_{<i}$ is
read. However, we insist on a worst-case time of $t_q$ (remember that
we allow Monte Carlo randomization, so we can put a worst-case bound
on the query time). Let $\foot_i(q)$ be the set of cells that the
query algorithm reads in this, possibly bogus, simulation.

Note that the decoder can compute $\foot_i(q)$ if it knows the cache
contents at query time and $W_{<i}$: it need only simulate $q$ and
assume optimistically that no cell from $W_i \setminus W_{<i}$ is
read.

For some set $X$, let $\foot_i(X) = \bigcup_{q\in X} \foot_i(q)$.  Now
define the footprint of epoch $i$: $F_i = \big( \foot_i(S_i) \cup W_i
\big) \setminus F_{<i}$. In other words, the footprint of an epoch
contains the cells written in the epoch and the cells that the
positive queries to that epoch would read in the decoder's optimistic
simulation, but excludes the footprints of smaller epochs. Note that
$F_{<i} = W_{<i} \cup \big( \bigcup_{j=0}^{i-1} \foot_j(S_j) \big)$.

We are now ready for the main part of our proof. Let $H_b(\cdot)$ be the
binary entropy.

\begin{lemma}  \label{lem:const-tu}
Let $i \in \{ \imin, \dots, \imax \}$ be good and $k = \lambda^i$. Let
$p$ be the probability over $\Dno$ that the query reads a cell from
$F_i$. We can encode a vector of $2k$ bits with $k$ ones by a message
of expected size: $O(\lambda^{i-1} t_q \cdot B \lg n) + 2k \cdot
\big[ H_b(\frac{1-\eps}{2}) + H_b(\frac{p}{2}) + H_b(3C_\eps) \big]$.
\end{lemma}

Before we prove the lemma, we show it implies the desired lower
bound. Choose $\lambda = B\lg^2 n$, which implies $\lambda^{i-1} t_q
\cdot B\lg n = o(\lambda^i) = o(k)$. Note that $H_b(\frac{1-\eps}{2})$
is a constant bounded below $1$ (depending on $\eps$). On the other
hand $H_b(3C_\eps) = O(C_\eps \lg \frac{1}{C_\eps})$. Thus, there
exist a small enough $C_\eps$ depending on $\eps$ such that
$H_b(1-\eps) + H_b(3C_\eps) < 1$. Since the entropy of the message is
$2k-O(\lg k)$, we must have $H_b(1-\eps) + H_b(3C_\eps) + H_b(p) \ge 1
- o(1)$. Thus $H_b(\frac{p}{2}) = \Omega(1)$, so $p=\Omega(1)$.

We have shown that a query over $\Dno$ reads a cell from $F_i$ with
constant probability, for any good $i$. Since half the $i$'s are good
and the $F_i$'s are disjoint by construction, the expected query time
must be $t_q = \Omega(\log_\lambda n) = \Omega(\log_B n)$.

\paragraph{Proof of Lemma \ref{lem:const-tu}.}
The encoding will consist of:
\begin{enumerate*}
\item The cache contents right before the query. 

\item The address and contents of cells in $F_{<i}$. In particular
  this includes $W_{<i}$, so the decoder can compute $\foot_i(q)$ for
  any query.

\item Let $X\subseteq Q\setminus S_i$ be the set of negative queries
  that read at least one cell from $F_i$. We encode $X$ as a subset of
  $Q$, taking $\log_2 \binom{2k}{|X|} + O(\lg k)$ bits. 

\item For each cell $c \in W_i \setminus F_{<i}$, mark some
  \emph{positive} query $q\in S_i$ such that $c \in \foot_i(q)$. Some
  cells may not mark any query (if they are not in the footprint of
  any positive query), and multiple cells may mark the same
  query. Let $M$ be the set of marked queries. We encode $M$ as a
  subset of $Q$, taking $\log_2 \binom{2k}{|M|} \le \log_2
  \binom{2k}{|W_i|}$ bits.

\item The set of queries of $Q$ that return incorrect results (Monte
  Carlo error).
\end{enumerate*}

\noindent
The decoder immediately knows that queries from $X$ are negative, and
queries from $M$ are positive. Now consider what happens if the
decoder simulates some query $q \in Q \setminus (X \cup M)$. If $q$
reads some cell from $\foot_i(M) \setminus F_{<i}$, we claim it must
be a positive query. Indeed, $M \subset S_i$, so $\foot_i(M) \setminus
F_{<i} \subseteq F_i$. But any negative query that reads from $F_i$
was identified in $X$.

\begin{claim}
If a query $q \in Q \setminus (X \cup M)$ does not read any cell from
$\foot_i(M) \setminus F_{<i}$, the decoder can simulate it correctly.
\end{claim}

\begin{proof}
We will prove that such a query does not read any cell from $W_i
\setminus F_{<i}$, which means that the decoder knows all necessary
cells. First consider positive $q$. If $q$ reads a cell from $W_i
\setminus F_{<i}$ and $q$ is not marked, it means this cell marked
some other $q'\in M$. But that means the cell is in $\foot_i(q')
\subseteq \foot_i(M)$, contradiction.

Now consider negative $q$. Note that $W_i \setminus F_{<i} \subseteq
F_i$, so if $q$ had read a cell from this set, it would have been
placed in the set $X$.
\end{proof}

Of course, some queries may give wrong answers when simulated
correctly (Monte Carlo error). The decoder can fix these using
component 5 of the encoding.

It remains to analyze the expected size of the encoding. To bound the
binomial coefficients, we use the inequality: $\log_2 \binom{n}{m} \le
n \cdot H_b(\frac{m}{n})$. We will also use convexity of $H_b(\cdot)$
repeatedly.
\begin{enumerate*}
\item The cache size is $M \le \lambda^{i-1}$ bits, since $i \ge
  \imin$. 

\item We have $\E[|F_{<i}|] \le \E[|W_{<i}|] + |S_{<i}|\cdot t_q \le
  O(\lambda^{i-1} t_q)$. So this component takes $O(\lambda^{i-1} t_q
  \cdot B\lg n)$ bits on average.

\item Since $\E[|X|] = pk$, this takes $2k \cdot H_b(\frac{p}{2})$
  bits on average.

\item Since $\E[|W_i|] = t_u k \le (1-\eps) k$, this takes $2k \cdot
  H_b(\frac{1-\eps}{2})$ bits on average.

\item The probability of an error is at most $2C_\eps$ on $\Dno$ and
  $4C_\eps$ on $\calD_i$. So we expect at most $6 C_\eps \cdot k$
  wrong answers. This component takes $2k \cdot H_b(3 C_\eps)$ bits on
  average.
\end{enumerate*}
\noindent
This completes the proof of Lemma~\ref{lem:const-tu}, and our analysis
of update time $t_u = 1-\eps$.

\subsection{Subconstant Update Time}

The main challenge in the constant regime was that we couldn't
identify which cells were the ones in $W_i$; rather we could only
identify the queries that read them. Since $t_u \ll 1$ here, we will
be able to identify those cells among the cells read by queries
$Q$, so the footprints are no longer a bottleneck.

The main bottleneck becomes writing $W_{<i}$ to the encoding.
Following the trick of \cite{patrascu07bit}, our message will instead
contain $W_i \cap R_{<i}$ and the cache contents at the end of epoch
$i$. We note that this allows the decoder to recover $W_{<i}$. Indeed,
the keys $S_{<i}$ are chosen by public coins, so the decoder can
simulate the data structure after the end of epoch $i$. Whenever the
update algorithm wants to read a cell, the decoder checks the message
to see if the cell was written in epoch $i$ (whether it is in $W_i
\cap R_{<i}$), and retrieves the contents if so.

We bound $W_i \cap R_{<i}$ by the following, which can be seen as a
strengthening of Observation~\ref{obs:wi}:
\begin{lemma}   \label{lem:riwi}
At least a quarter of the choices of $i \in \{ \imin, \dots, \imax \}$
are good and satisfy $\E[ | W_i \cap R_{< i} |] = O(\lambda^{i-1} t_u
/ \log_\lambda n)$.
\end{lemma}

\begin{proof}
We will now choose $i$ randomly, and calculate $\E[ |W_i \cap R_{<i}|
  / \lambda^i]$, where the expectation is over the random distribution
and random $i$. We will show $\E[ |W_i \cap R_{<i}| / \lambda^i] =
O(t_u / (\lambda \log_\lambda n))$, from which the lemma follows by a
Markov bound and union bound (ruling out the $i$'s that are not good).

A cell is included in $W_i \cap R_{<i}$ if it is read at some time $r$
that falls in epochs $\{0, \dots, i-1\}$, and the last time it was
written is some $w$ that falls in epoch $i$. For fixed $w<r$, let us
analyze the probability that this happens, over the random choice of
the query's position. There are two necessary and sufficient
conditions for the event to happen:
\begin{itemize*}
\item the boundary between epochs $i$ and $i+1$ must occur before $w$,
  so the query must appear before $w + \sum_{j\le i} \lambda^j < w +
  2\lambda^i$. Since the query must also appear after $r$, the event
  is impossible unless $r-w < 2\lambda^i$.
  
\item the boundary between epochs $i$ and $i-1$ must occur in $(w,r)$,
  so there are at most $r-w$ favorable choices for the query. Note
  also that the query must occur before $r + \sum_{j<i} \lambda^j$, so
  there are at most $2\lambda^{i-1}$ favorable choices.
\end{itemize*}
Let $i^\star$ be the smallest value such that $r-w <
2\lambda^{i^\star}$.  Let us analyze the contribution of this
operation to $\E[ |W_i \cap R_{<i}| / \lambda^i]$ for various $i$.  If
$i < i^\star$, the contribution is zero. For $i = i^\star$, we use
the bound $2\lambda^{i-1}$ for the number of favorable choices. Thus,
the contribution is $O(\frac{\lambda^{i-1}}{n} \cdot \lambda^{-i}) =
O(\frac{1}{n\lambda})$. For $i > i^\star$, we use the bound $r-w$ for
the number of favorable choices. Thus, the contribution is
$O(\frac{r-w}{n} \cdot \lambda^{-i}) = O(\frac{1}{n} \lambda^{i^\star
  - i})$.

We conclude that the contribution is a geometric sum dominated by
$O(\frac{1}{n\lambda})$. Overall, there are at most $n \cdot t_u$
memory reads in expectation, so averaging over the choice of $i$, we
obtain $\frac{1}{\log_\lambda n} \cdot O(\frac{1}{n\lambda}) \cdot
nt_u = O(\frac{t_u}{\lambda \log_{\lambda} n})$.
\end{proof}

Our main claim is:

\begin{lemma}   \label{lem:subconst-tu}
Let $i \in \{ \imin, \dots, \imax \}$ be chosen as in
Lemma~\ref{lem:riwi} and $k = \lambda^i$. Let $p$ be the probability
over $\Dno$ that the query reads a cell from $W_i \setminus
W_{<i}$. We can encode a vector of $2k$ bits with $k$ ones by a
message of expected size: $O(\lambda^{i-1} t_u \cdot B \lg \lambda) +
O(k t_u \lg \frac{t_q}{t_u}) + 2k \cdot \big[ H_b(\frac{p}{2}) +
  H_b(3C_\eps) \big]$.
\end{lemma}

We first show how the lemma implies the desired lower bound. Let
$\lambda$ be such that the first term in the message size is $\le
\lambda^i = k$. This requires setting $\lambda = O(t_u B \lg (t_u B))$.

Note that the lower bound is the same for $t_u=1-\eps$ and, say, $t_u
= 1/\sqrt{B}$. Thus, we may assume that $t_u\le 1/\sqrt{B} \le
1/\sqrt{\lg n}$. Therefore $k t_u \lg \frac{t_q}{t_u} = O(k \cdot
\frac{\lg\lg n}{\sqrt{\lg n}}) = o(k)$.

Since the entropy is $2k-o(\lg n)$, we must have $\frac{1}{2} +
H_b(\frac{p}{2}) + H_b(3C_\eps) \ge 1-o(1)$. For a small enough
$C_\eps$, we obtain $H_b(\frac{p}{2}) = \Omega(1)$, so $p= \Omega(1)$.

We have shown that a query over $\Dno$ reads a cell from $W_i
\setminus W_{<i}$ with constant probability, for at least a quarter of
the choices of $i$. By linearity of expectation $t_q =
\Omega(\log_\lambda n) = \Omega(\lg n / \lg (Bt_u))$.

\paragraph{Proof of Lemma \ref{lem:subconst-tu}.}
The encoding will contain:
\begin{enumerate*}
\item The cache contents at the end of epoch $i$, and right before the
  query.

\item The address and contents of cells in $W_i \cap R_{<i}$. The
  decoder can recover $W_{<i}$ by simulating the updates after epoch
  $i$.

\item Let $X\subseteq Q\setminus S_i$ be the set of negative queries
  that read at least one cell from $W_i \setminus W_{<i}$. We encode
  $X$ as a subset of $Q$. 

\item For each cell $c \in W_i \setminus W_{<i}$, mark some positive
  query $q\in S_i$ such that $c$ is the first cell from $W_i \setminus
  W_{<i}$ that $q$ reads. Note that distinct cells can only mark
  distinct queries, but some cells may not mark any query if no
  positive query reads them first among the set.  Let $M$ be the set
  of marked queries, and encode $M$ as a subset of $Q$.

\item For each $q\in M$, encode the number of cell probes before the
  first probed cell from $W_i \setminus W_{<i}$, using $O(\lg t_q)$
  bits.

\item The subset of queries from $Q$ that return a wrong answer (Monte
  Carlo error).
\end{enumerate*}

\noindent
The decoder starts by simulating the queries $M$, and stops at the
first cell from $W_i \setminus W_{<i}$ that they read (this is
identified in part 5 of the encoding). The simulation cannot continue
because the decoder doesn't know these cells. Let $W^\star \subseteq
W_i \setminus W_{<i}$ be the cells where this simulation stops.

The decoder knows that queries from $M$ are positive and queries from
$X$ are negative. It will simulate the other queries from $Q$. If a
query tries to read a cell from $W^\star$, the simulation is stopped
and the query is declared to be positive. Otherwise, the simulation is
run to completion.

We claim this simulation is correct. If a query is negative and it is
not in $X$, it cannot read anything from $W_i \setminus W_{<i}$, so
the simulation only uses known cells. If a query is positive but it is
not in $M$, it means either: (1) it doesn't read any cell from $W_i
\setminus W_{<i}$ (so the simulation is correct); or (2) the first
cell from $W_i \setminus W_{<i}$ that it reads is in $W^\star$,
because it marked some other query in $M$. By looking at component 6,
the decoder can correct wrong answers from simulated queries.

It remain to analyze the size of the encoding:
\begin{enumerate*}
\item The cache size is $M = O(\lambda^{i-1})$ bits.

\item By Lemma~\ref{lem:riwi}, $\E[|W_i \cap R_{<i}|] =
  O(\lambda^{i-1} \frac{t_u} {\log_\lambda n})$, so this component
  takes $O(\lambda^{i-1} t_u \log \lambda \cdot B)$ bits.

\item Since $\E[|X|] \le p\cdot k$, this component takes $\E[\log_2
  \binom{2k}{|X|}] \le 2k \cdot H_b(\frac{p}{2})$ bits.

\item Since $\E[|W_i|] \le \lambda^i t_u = kt_u$, this takes $\E[ \log_2
  \binom{2k}{|W_i|} ] \le k \cdot O(t_u \lg \frac{1}{t_u})$ bits.

\item This takes $\E[|W_i|] \cdot O(\lg t_q) = O(k t_u \lg t_q)$ bits.

\item We expect $2C_\eps\cdot k$ wrong negative queries and $4C_\eps
  \cdot k$ wrong positive queries, so this component takes $2k \cdot
  H_b(3 C_\eps)$ bits on average.
\end{enumerate*}

\noindent
This completes the proof of Lemma~\ref{lem:subconst-tu}, and our
lower bound.

\end{document}